\documentclass[letterpaper,10pt,conference]{ieeeconf}
\IEEEoverridecommandlockouts      
\author{ Shuang~Gao and  Peter E. Caines
\thanks{*This work is supported in part by NSERC (Canada) Grant RGPIN-2024-06612 (SG), Simons-Berkeley Research Fellowship (SG)  and NSERC (Canada) Grant RGPIN- 2019-05336 (PEC).}
\thanks{Shuang Gao is with the Department of Electrical  Engineering, Polytechnique Montreal \& GERAD,
   Montreal, QC, Canada. 
      Email:   {\tt\small shuang.gao@polymtl.ca}.
      Peter E. Caines is with the Department of Electrical and Computer Engineering, McGill University \& GERAD, Montreal, QC, Canada. Email:   {\tt\small peterc@cim.mcgill.ca}.
      }%
\thanks{The authors would like to thank  Christian Borgs, George Cantwell, Jorge Velasco-Hern\'andez, Xin Guo, Minyi Huang, Aditya Mahajan, Charlotte Govignon, and reviewers  for their helpful feedback.
}  
}%
\let\labelindent\relax
\usepackage{enumitem}
\setlist[description]{leftmargin=2.5em, labelsep=0.5em, style=standard}

\usepackage[noadjust]{cite}

\usepackage{graphics}
\usepackage{epstopdf}

\usepackage{amsmath}
  \usepackage[font=footnotesize]{subfig}

\usepackage{hyperref}
\usepackage{enumerate}
\usepackage{color}

\newcommand{\BN}{\mathds{N}}

\newcommand*\TRANS{{\mathpalette\doTRANS\empty}}
\makeatletter
\newcommand*\doTRANS[2]{\raisebox{\depth}{$\m@th#1\intercal$}}
\makeatother

\usepackage{url}
\usepackage{amssymb,mathtools}

\usepackage[standard,thmmarks,amsmath]{ntheorem}

\usepackage[T1]{fontenc}
\usepackage{xcolor}
\usepackage{mathtools, amssymb, amsfonts, dsfont, nccmath}
\usepackage{microtype}

\usepackage{times}
\usepackage{dsfont}
\usepackage{tikz}

\begin{document}
%
\title{Transmission Neural Networks: \\
Approximate Receding Horizon Control for Virus Spread on Networks}
%
%
%
\maketitle 
\begin{abstract}                
Transmission Neural Networks (TransNNs) proposed by Gao and Caines (2022) serve as both virus spread models  over networks and neural network models with tuneable activation functions. 
This paper establishes that TransNNs provide upper bounds on the infection probability generated from the associated   Markovian  stochastic Susceptible-Infected-Susceptible (SIS) model with $2^n$ state configurations where $n$ is the number of nodes in the  network, and can be employed as an approximate model for the latter. %
Based on such an approximation, 
 a TransNN-based receding horizon control approach for mitigating virus spread is proposed and we demonstrate that it allows significant computational savings compared to the dynamic programming solution to Markovian SIS model with $2^n$ state configurations, as well as providing less conservative control actions compared to the TransNN-based optimal control.  Finally,  numerical comparisons among (a) dynamic programming solutions for the Markovian SIS model, (b) TransNN-based optimal control  and  (c) the proposed TransNN-based receding horizon control  are presented.
\end{abstract}


\section{Introduction}

Epidemic models are important in  predicting and mitigating epidemic spreads, and  many different epidemic models have been proposed and analyzed (see~\cite{pastor2015epidemic,nowzari2016analysis, kiss2017mathematics, van2014performance,liggett2013stochastic,pare2020modeling}).
For epidemic spread models over networks, 
a first thorough system-theoretic analysis was presented in~\cite{lajmanovich1976deterministic}, where the network characterizes transmission probability rates among different population groups.  
 Virus spread processes on random directed graphs were analyzed in~\cite{kephart1992directed}, and virus spread models on  networks characterized by degree distributions were studied  in~\cite{pastor2001epidemic}.  %
Mean field approximations  for continuous-time SIS virus spread models on networks have been developed in~\cite{van2008virus,cator2012second,van2015accuracy,van2014performance}, where mean field states approximate the fractions of the infected in nodal populations.  Continuous-time contact processes can  also model epidemic spreads (see e.g.~\cite{liggett2013stochastic}). %

Discrete-time SIS models over networks, closely related to the current paper, have been investigated by various researchers.  Discrete-time SIS models  over networks with homogeneous transmission probabilities have been proposed for identifying threshold values for epidemics spread in~\cite{wang2003epidemic} and~\cite{chakrabarti2008epidemic}. 
The nonlinear discrete-time SIS model in~\cite{chakrabarti2008epidemic} provides an upper bound on the probability of infection generated from the discrete-time Markov chain model with $2^n$ state configurations with $n$ as the number of nodes, and such a result was established for the case with homogeneous infection probabilities across all links in~\cite{ahn2014mixing}. In~\cite{ahn2013global}, the model from~\cite{chakrabarti2008epidemic} was linearized to  provide upper bounds for infection states, and the stability was then analyzed for the linearized model.  
%
%
%
The work~\cite{han2015data} extended the discrete-time SIS model from~\cite{wang2003epidemic, chakrabarti2008epidemic} to the case with non-homogeneous transmission probabilities, obtained a linear dynamic model that provides an upper bound for the probability of infection, and  then used the linear model to solve vaccine allocation problems via geometric programming. In the model proposed in~\cite{paarporn2015epidemic}, each agent  makes a social interaction decision based  on a local awareness  that depends on states of other agents, and an  SIS model approximating the stochastic SIS model over networks with $2^n$ state configurations was analyzed, which is a different approximation model from the current paper. An observer model was proposed for the discrete-time stochastic SIS model with  $2^n$ state configurations for designing feedback control  in~\cite{watkins2017inference}.
A comparative analysis of two discrete-time SIS epidemic  models (where one is based on the Euler discretization of the continuous time SIS model in~\cite{van2008virus} and the other is a variant of the discrete time SIS model in~\cite{chakrabarti2008epidemic,han2015data,ahn2014mixing}) was carried out in~\cite{pare2018analysis}. %
 Both models converge to the continuous time SIS model in~\cite{van2008virus} with infinitesimal sampling time (see~\cite{pare2018analysis,ShuangPeterTransNN22}). The connection between discrete-time SIS models with non-homogeneous transmission probabilities and neural networks was established in~\cite{ShuangPeterTransNN22} via TransNNs. In~\cite{ShuangPeterTransNNControl25}, it was shown that TransNNs enable (approximate) optimal control for virus spreads with significant computational reductions compared to dynamic programming solutions for the Markovian SIS epidemic models.  Compared to earlier discrete-time virus spread models (e.g.~\cite{wang2003epidemic, chakrabarti2008epidemic,ahn2013global,ahn2014mixing,pare2018analysis}), TransNNs offer more concise representations of SIS dynamics over networks. In addition, TransNNs  simplify the stability analysis for SIS models over networks where the  transmission probabilities and connections are potentially heterogenous and time-varying~\cite{ShuangPeterTransNN22}.  All the papers discussed above differ from the current work either in terms of  the models or the results.%

 As learning models,  TransNNs are universal function approximators and  contain trainable  activation functions~\cite{ShuangPeterTransNN22}.    
The key conceptual difference between TransNNs and standard neural networks is that trainable activation functions in TransNNs  are associated with links, whereas in standard neural networks,  
activation functions are typically considered as a part of nodes (see~\cite{ShuangTransNNVideo22} for detailed explanations). {Such a feature with trainable link functions also appeared in~\cite{liu2024kan}.} %

{Main contributions of the current paper are as follows. Firstly, the relation between the Markovian stochastic SIS  model with $2^n$ state configurations (over possibly heterogeneous and time-varying transmission probabilities) and the  TransNN model is established; more specifically, we prove that TransNNs provide upper bounds for the probabilities of infection generated from the stochastic SIS epidemic model  under mild technical assumptions. 
Secondly, we demonstrated that TransNNs enable receding horizon control formulations to mitigate virus spread over  (time-varying and random) networks and allow significant computational reduction compared to the control solutions based on stochastic SIS epidemic models with $2^n$ state configurations. 
}

\section{Virus Spread Dynamics Over Networks}
Consider a physical (directed) contact network denoted by $(V,E^k)$ with an adjacency matrix $A_k=[a_{ij}^k]$, where $V$ is the node set and $E^k\subset V\times V$ is the edge set at time $k$. Let $(i,j) \in E^k$ denote the direct edge from node $j$ to node $i$ at time $k$.
 Each node of the physical contact network  may represent an individual person and a link between two persons exists, for instance, if they are within a given distance for an extended period of time (e.g. within 2 meters for at least 15mins). 
For all $i,j\in V$, let $w_{ij}^k$ denote the probability of node $j$ infecting its neighbouring node $i$ on the physical contact network given that $j$ is infected at time $k$, and  let $W_{ij}^k \in\{0,1\}$ be the binary random variable representing the successful transmission of virus from node $j$ to node $i$ at  time $k$. For simplicity, let $V =[n]\triangleq \{1,2,..., n\}$.

The state of a node $i \in [n]$ at time $k$ is denoted by a random variable $X_i(k)$, which takes binary values $0$ and $1$ with $0$ representing the healthy state and $1$ representing the infected state.  The one-step update of the  binary random states for epidemic spread over networks given the current states $(X_i(k))_{i\in [n]}$ follows the dynamics
\begin{equation}\label{eq:agent-based}
	1- X_i(k+1) = \prod_{j\in {N}_i^{\circ k}} \big(1- W_{ij}^kX_j(k)\big), \quad \forall i \in [n]
\end{equation}
{where $N_i^{\circ k}  \triangleq  \{j: (i,j)\in E^k \} \cup \{i\}$ denotes the (incoming) neighbourhood set of node $i$ with itself included at time $k$, and $(V, E^k)$ denotes the physical contact graph at time $k$. 
 {The dynamics have the property that for a healthy node to become infected, the infection needs to come from at least one of its neighbours, and furthermore the successful transmission of the virus from  one neighbour of node $i$ is sufficient to infect node $i$.} 
 It is worth highlighting that the inclusion of  self-loops (with the neighborhood represented by $N_i^{\circ k} $ for node $i$ at time $k$)  is essential in the characterization of the recovery process  in the virus spread dynamics, as the recovery process is equivalently represented by the self-transmissions with self-loops.
 
Let $X(k) \triangleq [X_1(k)\dots X_n(k)]^\TRANS \in \{0,1\}^n $ denote the state configuration at time $k$ and let $W^k\triangleq [W_{ij}^k] \in \{0,1\}^{n\times n}$. 
We introduce the following assumptions. 
\begin{description}
\item[\bf(A1)] At any time $k> 0$,	 $W^k\triangleq [W_{ij}^k]$ is independent of $\{W^t: 0\leq t < k\}$ and $\{X{(t)}: 0\leq t < k \}$. 
\item[\bf (A2)]  At any time $k\geq 0$, the  binary random variables $\{W_{ij}^k:i,j \in [n]\}$ representing the transmissions are  conditionally (jointly) independent given the current state configuration $X(k)$. 
\end{description}

The assumption (A1) ensures that transmissions are independent of the history of states and past transmissions, and (A2) ensures that transmissions are conditionally independent given the current infection states.

Under (A1), the dynamics \eqref{eq:agent-based} are Markovian, since 
\begin{equation}
\begin{aligned}
	 \mathbb{E}(1-& X_i(k+1)| (X(t))_{t\in [k]}) \\
	& = \mathbb{E}\bigg(\prod_{j\in {N}_i^{\circ k}} \big(1- W_{ij}^kX_j(k)\big)\Big| (X(t))_{t\in [k]} \bigg)\\
		& = \mathbb{E}\bigg(\prod_{j\in {N}_i^{\circ k}} \big(1- W_{ij}^kX_j(k)\big)\Big| X(k)\bigg)\\
	& =  \mathbb{E}(1- X_i(k+1)| X(k)).
\end{aligned}
\end{equation} 
Furthermore, under both  (A1) and (A2), 
we have  
\begin{equation*}
\begin{aligned}
\mathbb{E}(1- X_i(k+1)|X(k)) & = \mathbb{E}\prod_{j\in {N}_i^{\circ k} } \Big(1- W_{ij}^kX_j(k)|X(k)\Big) \\
& = \prod_{j\in {N}_i^{\circ k} }\mathbb{E}  \Big(1- W_{ij}^kX_j(k)|X(k)\Big)
\end{aligned}
\end{equation*} 
which corresponds to the virus spread model (with one-step update) proposed in~\cite{ShuangPeterTransNN22}, since  the conditional probability satisfies
 \begin{equation}%
1-	\bar{\rho}_i(k+1) = \prod_{j\in {N}_i^{\circ k}} \Big(1- w_{ij}^k\bar{\rho}_j(k)\Big), \quad i\in [n],
\end{equation}
where  
$
w_{ij}^k\triangleq \textup{Pr}(W_{ij}^k=1 {| X_j(k)=1})
$ and  $$\bar{\rho}_i(k)\triangleq \mathbb{E}(X_i(k)|X(k-1)) = \text{Pr}(X_i(k)=1|X(k-1)).$$ 
For a state configuration $q\in\{0,1\}^n$, the probability of reaching $q$ is given by
\begin{equation*} 
\begin{aligned}
	\text{Pr}(&X(k+1)=q|X(k))   = \prod_{i=1}^n \text{Pr}(X_i(k+1)=q_i|X(k)) 
\end{aligned}
\end{equation*}
where the equality is due to the conditional independence of the virus transmissions $\{ W_{ij}^k\}$ assumed in (A2). %

\begin{proposition}[Conditional Probability of Infection~\cite{ShuangPeterTransNNControl25}] \label{prop:cond-prob-infection}
Assume (A1) and (A2) hold. 
Given the state configuration at time $k$ denoted by ${X(k)= x} \in \{0,1\}^n$, the transition probability to a state configuration $q\in\{0,1\}^n$ is  given by 
\begin{equation}\label{eq:MarkovTranProb2}
\begin{aligned}
	\textup{Pr}(X(&k+1)=q|X(k)=x)\\ & = \prod_{i=1}^n \textup{Pr}(X_i(k+1)=q_i|X(k)= x)\\
	& = \prod_{i=1}^n \Big(q_i\rho_i(k+1) + (1-q_i) (1-\rho_i(k+1)) \Big)
\end{aligned}
\end{equation}
where
 \begin{equation}\label{eq:prob-configuration}
\rho_i(k+1) = 1-	\prod_{j\in {N}_i^{\circ k}} \Big(1- w_{ij}^k x_j\Big), \quad i\in [n].
\end{equation}
\end{proposition}
This explicit representation of the transition probability to the next state configuration given the current state configuration will be used later in the Markov Decision Process (MDP) formulation for controlling virus spread over networks (see Section \ref{sec:control-virus-spread}). 

\section{Approximation by TransNNs}
To derive TransNNs that approximate the virus spread dynamics \eqref{eq:agent-based}, we  introduce the following  assumptions.  %
\begin{description}
	\item[\bf (A3)]  At time $k\geq 0 $, $\{ W_{ij}^k: i,j \in [n]\}$ are  independent and  for each $i, j \in [n]$, $W_{ij}^k$ is independent of  $\{ X_\ell(k): \ell\in [n],  \ell\neq j\}$. 
	\item[\bf (A4)] The  events $\{\{X_i(k)=1\}:i\in [n], 0\leq k\leq T \}$ are  independent, with $T$ as the terminal time.
\end{description}
The assumption (A3)  introduces the independence of transmissions over different links at the current time $k$ and the independence of transmissions with respect to the states except the current state at the sending node.  Hence the assumption (A3) is stronger than the assumption (A2).  

Assumptions {(A1)}, (A3) and (A4) together  allow us to write the update of the expected state by taking the (total) expectation inside the product on the right-hand side  as follows: for $i\in [n]$, 
\begin{equation}\label{eq:naive-mf-approx}
	\begin{aligned}
\mathbb{E}(1- X_i(k+1)) & = \mathbb{E}\prod_{j\in {N}_i^{\circ k} } \Big(1- W_{ij}^kX_j(k)\Big) \\
& = \prod_{j\in {N}_i^{\circ k} }\mathbb{E}  \Big(1- W_{ij}^kX_j(k)\Big) .
\end{aligned}
\end{equation}
Under  {(A1)}, (A3) and (A4), the joint probability distribution for the random variables $\{Z_{ij}(k )\triangleq 1- W_{ij}^kX_j(k), {j\in {N}_i^{\circ k} } \}$ on the right-hand side is equal to the product of its marginal distributions.  
If (A3) and (A4) are not satisfied, the right-hand side of \eqref{eq:naive-mf-approx} constitutes an approximation of the joint distribution by the product of the marginals,  referred to as the naive mean-field approximation~\cite{barber2012bayesian} or individual-based mean-field approximation~\cite{van2008virus}. %

Let $w_{ij}^k\triangleq \operatorname{Pr}(W_{ij}^k=1 {| X_j(k)=1})$ denote the conditional probability of the successful transmission of the virus from node $j$ to node~$i$ at time step $k$. %
Then
\[
\mathbb{E}  \big(1- W_{ij}^kX_j(k)\big) =  \big(1-   w_{ij}^k \mathbb{E}  X_j(k)\big), \quad \forall i, j \in [n]. 
\] 
Hence under (A1),  (A3) and (A4), the dynamics in \eqref{eq:naive-mf-approx} lead to the  virus spread model in~\cite{ShuangPeterTransNN22}:
\begin{equation}\label{eq:TransNNs-virus}
1-	p_i(k+1) = \prod_{j\in {N}_i^{\circ k}} \big(1- w_{ij}^kp_j(k)\big), \quad i\in [n],
\end{equation}
with  $p_i(k)\triangleq \textup{Pr}(X_i(k)=1) = \mathbb{E}X_i(k) $ representing the (total) probability of node $i$ being infected at time $k$.
Furthermore, by the nonlinear state transformation     
\[ 
	s_i(k)\triangleq -\log(1-p_i(k)) ~ \in [0,\infty],
\]  
proposed in~\cite{ShuangPeterTransNN22},
the equivalent TransNN representation of \eqref{eq:TransNNs-virus} is then given by  
\begin{equation}\label{eq:TransNN}
	s_i(k+1)=\sum_{j\in {N}_i^{\circ k}} \Psi(w_{ij}^k, s_j(k)), \quad i \in [n]
\end{equation}
with the TlogSigmoid activation function~\cite{ShuangPeterTransNN22} 
\begin{equation}\label{eq:TlogSigmoid}
\Psi(w,x)= -\log(1-w+we^{-x})	
\end{equation}
 with the activation level $w\in[0,1]$ and the input signal $x\in[0,+\infty]$. Interested readers are referred to~\cite{ShuangPeterTransNN22,ShuangTransNNVideo22} for detailed properties of the TransNN model and its connections with standard neural network models.

 Thus, in  the stochastic model with $2^n$ state configurations described in \eqref{eq:agent-based}, imposing the independence assumptions (A3) and (A4) together with (A1)  leads to the virus spread model in  \eqref{eq:TransNNs-virus} and the equivalent TransNN model  \eqref{eq:TransNN}. 

\begin{remark}[Discussions on Assumption (A4)]
{Depending on the actual paths of infection over the networks, the assumption (A4) may not hold.}
There are two ways to get around the assumption on independence in (A4). One way to break the potential dependence among  $\{X_i(k)=1,i\in [n]\}$ is to  observe the states at each time step (i.e. setting the terminal time in (A4) to $T=1$ and re-initializing the dynamics at every time step with the new observations). The disadvantage of this approach is that we need to observe or simulate the states at each time step.  A similar idea is  used in the context of Restricted Boltzmann Machines (RBMs)~\cite{salakhutdinov2007restricted} if we interpret the time steps in the current paper as layers of bipartite graphs in RBMs. 
Alternatively,  we can sample points from the joint probability distribution to approximate and simulate the evolutions of the empirical distributions. The disadvantage is that  we must track the evolution of the empirical joint distribution over a state space with $2^n$ possible state configurations at each time.

Both approaches may incur significant computational burden, especially when the underlying network is large and complex.
Instead, the model in \eqref{eq:TransNNs-virus} can predict approximately the infection probabilities with lower complexity, {using only the infection status (or probabilities) at the initial time}. In fact, such an approximate model enables the control with low complexity, which will be presented later in Section~\ref{sec:control-virus-spread}.  	 
\end{remark}	 

\section{Infection Probability Upper Bound by TransNNs}
In this section,  we investigate further the relation between the state evolution in \eqref{eq:agent-based} and that  in \eqref{eq:TransNNs-virus}; more specifically, we show that without the assumption (A4) the states of the  dynamics in \eqref{eq:TransNNs-virus} (and equivalently those of TransNNs in \eqref{eq:TransNN}) provide upper bounds for the infection probabilities from the stochastic virus spread model in \eqref{eq:agent-based}. 
\begin{proposition}
\label{prop:prob-upper-bound}
Assume (A1) and (A3) hold. 
Let $X_i(k),$ $i \in[n]$ be the binary state in dynamics \eqref{eq:agent-based} and let $p_i(k)$ denote the probability state in  dynamics \eqref{eq:TransNNs-virus}. Given the same physical contact network $(V, E^k)$ for all  $k \in \mathds{N}$, the same conditional transmission probabilities $\{w_{ij}^k, i,j \in [n], k\geq 0\}$ and the same initial infection probability (i.e. $p_i(0)=\operatorname{Pr}(X_i(0)=1)$ for all $i\in [n]$), the following inequality holds
\begin{equation}\label{eq:prop2}
	p_i(k) \geq \operatorname{Pr}(X_i(k)=1), \quad \forall i \in [n], ~~\forall k \in \BN;
\end{equation}
if, furthermore, \textup{(A4)} is satisfied for the dynamics \eqref{eq:agent-based},  then the inequality in \eqref{eq:prop2} becomes an equality  for all $k\in\{0,1,\dots, T\}$ and  all $i \in [n]$.
\end{proposition}

\begin{proof}
Define the binary random variable $Z_{ij}(k) \triangleq 1-W_{ij}^kX_j(k)$.  Then  $Z_{ij}(k)=1$ represents the event that node $i$ does not successfully receive virus transmission from node~$j$  at time $k$.
We observe that for any~$i \in [n]$,
 $$
 \begin{aligned}
 	 \mathbb{E}&(1- X_i(k+1))  \\
 	 & = \mathbb{E}\prod_{j\in {N}_i^{\circ k}} \Big(1- W_{ij}^kX_j(k)\Big)  \triangleq \mathbb{E}\prod_{j\in {N}_i^{\circ k}} Z_{ij}(k) \\
&  = \text{Pr} \big( (Z_{ij}=1)_{j\in {N}_i^{\circ k}}\big) 
\geq \prod_{j\in {N}_i^{\circ k}}\text{Pr}(Z_{ij}(k)=1) \\
 \end{aligned}
$$
since the following holds
$$
\text{Pr} ( Z_{ij}(k)=1| (Z_{i\ell}(k)=1)_{\ell \in S\subset [n], \ell \neq j} ) \geq  \text{Pr} ( Z_{ij}(k)=1),
$$
for all subset $S\in [n]$ and all $i, j \in [n]$.
That is, conditioning on  node $i$ not successfully receiving  virus transmissions from  neighbours other than node $j$, does not reduce the probability that node $i$ does not successfully receive   virus transmission from node $j$\footnote{In fact, this is the case even when $X_i(k)$ for all $i\in[n]$ are not independent. A similar argument was used in~\cite{van2008virus}. 
Such an argument holds for Susceptible-Infected-Susceptible (SIS) virus spread models and it may not hold in general for Susceptible-Infected-Recovered (SIR)  models.}.   
Applying the property above, we obtain that 
 \begin{equation}
 	 \begin{aligned}
	 \mathbb{E}(1- X_i(k+1))  &= \mathbb{E}\prod_{j\in {N}_i^{\circ k}} \Big(1- W_{ij}^kX_j(k)\Big)  \\
&  \geq \prod_{j\in {N}_i^{\circ k}} \mathbb{E}\Big(1- W_{ij}^kX_j(k)\Big) \\
&  =  \prod_{j\in {N}_i^{\circ k}}\Big(1-  w_{ij} \mathbb{E} X_j(k)\Big)
 \end{aligned}
 \end{equation}
 where $w_{ij}\triangleq \text{Pr}(W_{ij}=1| X_j=1)$ denotes the conditional probability of the successful virus transmission  from node $j$ to node $i$ given that node $j$ is infected. 
Thus we obtain 
 \begin{equation}\label{eq:update-inequality}
 	 \mathbb{E}X_i(k+1) \leq 1- \prod_{j\in {N}_i^{\circ k}} \Big(1- w_{_{ij}} \mathbb{E}X_j(k)\Big),\quad \forall i \in [n].
 \end{equation}
Firstly, we note that this inequality is element-wise (i.e. it holds for each $i\in [n]$). Starting from the initial condition $p_i(0)= \mathbb{E}X_i(0)$ for all $i \in [n]$, we obtain $p_i(1)\geq \mathbb{E}X_i(1)$ for all $i \in [n]$. 
Secondly, we observe that the function 
\begin{equation} \label{eq:rhs-monotone}
f(y_1,..., y_n) \triangleq 1- \prod_{j\in {N}_i^{\circ k}} \big(1- w_{_{ij}} y_j\big)	
\end{equation}
 corresponding to the right-hand side of \eqref{eq:update-inequality} is monotonically increasing with respect to $y_j$ for all $j\in {N}_i^{\circ k}$. 
Thus, 
\[
\begin{aligned}
p_i(2) & = 1- \prod_{j\in {N}_i^{\circ k}} \Big(1- w_{_{ij}} p_j(1)\Big) \\
& \geq 
1- \prod_{j\in {N}_i^{\circ k}} \Big(1- w_{_{ij}} \mathbb{E}X_j(1)\Big)
\geq 
\mathbb{E}X_i(2),\quad \forall i \in [n].
\end{aligned}
\]
Thus by induction we obtain $p_i(k)\geq \mathbb{E}X_i(k)$ for all $i \in [n]$, for all $k\in \mathbb{N}$.
Finally, since the state $X_i\in \{0,1\}$ is binary, we obtain $\mathbb{E}X_i(k) = \operatorname{Pr}(X_i(k)=1)$ and hence   \eqref{eq:prop2}. 
\end{proof}

A similar inequality holds for Shannon information content (of being healthy). 
Consider the  state transformation (see~\cite{ShuangPeterTransNN22}):
$$
\mathbb{T}(p)= -\log(1-p) \in [0,\infty],\quad  \forall p\in[0,1].
$$
The state transformation above is bijective (from $[0,1]$ to $[0, \infty]$) and monotonically increasing with respect to the input. Taking $\mathbb{T}(\mathbb{E}X_i)$ as the state of node $i$, then the dynamics in \eqref{eq:update-inequality} leads to the following upper bound dynamics
\[
\mathbb{T}(\mathbb{E}X_i(k+1))\leq \sum_{j\in N_i^{\circ k}}  \Psi(w_{ij}, \mathbb{T}(\mathbb{E}X_i(k)))
\]
in terms of the evolution of the Shannon information content (of being healthy).
Then by Proposition~\ref{prop:prob-upper-bound}, we obtain that
\begin{equation}
	s_i(k)\geq \mathbb{T}(\mathbb{E}X_i)\triangleq -\log(1-\mathbb{E}X_i),
\end{equation}
that is, the states following the TransNN dynamics in \eqref{eq:TransNN} provide an upper bound for the information content of being healthy for node $i\in [n]$. Then taking the inverse mapping of $\mathbb{T}$ on both sides, we obtain the following result.
\begin{proposition}
\label{prop:prob-upper-bound-TransNNs}
Assume (A1) and (A3) hold. 
Let $X_i(k),$ $i \in[n]$ be the binary state in dynamics \eqref{eq:agent-based} and let $s_i(k)$ denote the  information content state in  dynamics \eqref{eq:TransNN}. 
Let $s_i(0) = -\log(1-\operatorname{Pr}(X_i(0)=1)) \in [0, +\infty]$ for all $i\in [n]$. 
 Given the same physical contact network $(V, E^k)$ for all  $k \in \mathds{N}$, the same conditional transmission probabilities $\{w_{ij}^k, i,j \in [n], k\geq 0\}$, the following inequality holds
\begin{equation}\label{eq:prop3}
	\operatorname{Pr}(X_i(k)=1) \leq 1-\exp(-s_i(k)), \quad \forall i \in [n], ~\forall k \in \BN;
\end{equation}
if, furthermore, \textup{(A4)} is satisfied for the dynamics \eqref{eq:agent-based},  then the inequality  in \eqref{eq:prop3} becomes an equality for all $k\in\{0,1,\dots, T\}$ and all $i \in [n]$. 
\end{proposition}

\begin{remark}
	{Under (A1) and (A3), the Markovian dynamics~\eqref{eq:agent-based} have an absorbing state configuration where all nodes are healthy. It implies that when the time runs sufficiently long, all nodes will eventually become healthy. However, when the number of nodes $n$ is large, the time required for the infection to end can be extremely long, and in such cases the expected mixing time can be analyzed instead, which depends on the network connectivity~\cite{ahn2014mixing}.}
\end{remark}

Since the TransNN model \eqref{eq:TransNN} and the equivalent virus spread model \eqref{eq:TransNNs-virus} are less conservative than the associated linear approximation of the SIS epidemic model (see the proof of~\cite[Thm. 1]{ShuangPeterTransNN22}),  simpler but more conservative upper bounds for the probability of infection can be obtained below.
\begin{proposition}[\cite{ShuangPeterTransNNControl25}]
	Let $X_i(k)$ be the binary state for node~$i$ at time~$k$ in the dynamics \eqref{eq:agent-based}. Let $A_k \triangleq  [a_{ij}^k]$ denote the adjacency matrix of the physical contact network and  $\Omega_k \triangleq  [w_{ij}^k]$ denote the matrix of the conditional probabilities of transmissions, at time $k$. Assume  \textup{(A1)} and \textup{(A3)} hold. Then for all $i \in [n]
$, the following holds
	\[
	 \operatorname{Pr}(X_i(k)=1) \leq [(A_k\odot \Omega_k)  \dots (A_0\odot \Omega_0) \mu(0)]_i 
	\]
	where the initial condition is  $\mu(0) \triangleq [\mu_1(0),\dots, \mu_n(0)]$, $\mu_i(0) \triangleq \operatorname{Pr}(X_i(0)=1)$, and $\odot$ is the Hadamard product. 
\end{proposition}

\section{Controlling Virus Spread  based on TransNNs} \label{sec:control-virus-spread}
Let the control $u_i(k) \in\{0,1\} $  represent the individual-level interventions that reduce the susceptibility of node $i$ to infection at time $k$; more specifically, at time $k$, $u_i(k) =1$ if individual $i$ undertakes protective measures (e.g. vaccination,  sanitization or  the use of protective equipment)   and $u_i(k)=~0$ otherwise.  
Assume the control (i.e. the protective measure)  reduces $1- \beta$ of the infection probability with $\beta \in [0,1]$. For a control action $u(k) \triangleq (u_1(k)\dots u_n(k))^\TRANS$, the controlled transmission probabilities are then given by
\begin{equation} 
m_{ij}^{k}(u_i(k)) \triangleq u_i(k) w_{ij}^k \beta + (1-u_i(k)) w_{ij}^k, ~ \forall i, j \in  [n].
\end{equation}

In the following, we first introduce MDP solutions for Markovian SIS dynamics, and then compare them  with control solutions based on TransNNs.
\subsection{MDP Solution for Markovian SIS Dynamics} \label{subsec:MDP-SIS}

Consider  the stochastic dynamics in \eqref{eq:agent-based}  with $w_{ij}^k$ replaced by the controlled transmission probability $m_{ij}^k(u_i(k))
$, and the cost below
\[
J_1 = \mathbb{E}\sum_{k=0}^{T-1} \mathbf{1}_n^\TRANS( c X(k) +   u(k))  \triangleq \mathbb{E}\sum_{k=0}^{T-1}  l (X(k), u(k))
\]
with transition probabilities given by \eqref{eq:MarkovTranProb2} and $\mathbf{1}_n$ as the $n$-dimensional vector of ones. In this case, the state space is $\{0,1\}^n$ and control action space is $\{0,1\}^n$. 
Let $V_k: \{0,1\}^n \to \mathbb{R}$ denote the value function defined by 
\[
V_k(x) \triangleq \min_{u_{k}\dots u_{T-1}} \mathbb{E}\sum_{\tau =k}^{T-1}  l (X(\tau), u(\tau)), \quad X(k) = x. 
\]
By Dynamic Programming (see e.g.~\cite{bertsekas2012dynamic}), the value function satisfies the Bellman equation
\[
V_k(x) = \min_{u} \Big[ l(x, u) + \sum_{x^\prime \in \{0,1\}^n}\text{Pr}(x^\prime|x, u ) V_{k+1}(x^\prime) \Big] 
\]
where $V_T =0$,  the transition probability following Proposition~\ref{prop:cond-prob-infection}  is specified by 
\begin{equation}
\begin{aligned}
	\text{Pr}(&X(k+1)=q|X(k)=x, u(k)=v)\\
	& = \prod_{i=1}^n \big(q_i\rho_i(k+1) + (1-q_i) (1-\rho_i(k+1)) \big)
\end{aligned}
\end{equation}
with
$
\rho_i(k+1) = 1-	\prod_{j\in {N}_i^{\circ k}} \big(1- m_{ij}^{k}(v_i)x_j(k)\big)
$
and \begin{equation} \label{eq:controlled-probability}
m_{ij}^{k}(v_i) = v_i w_{ij}^k \beta + (1-v_i) w_{ij}^k, ~ \forall i, j \in  [n].
\end{equation}
The optimal control is then given by the arguments of the minima on the right-hand side of the Bellman equation. 

\subsection{Optimal Control based on TransNNs} \label{subsec:oc-TransNNs}
Consider the same cost (adapted from $J_1$ \footnote{An application of Proposition \ref{prop:prob-upper-bound} yields that under the same control actions, if assumptions (A1) and (A3) are satisfied,
$J_1\leq J_2$ holds, and if furthermore the assumption (A4) is satisfied, $J_1 = J_2$ holds.})
\[
J_2 = 
\sum_{k=0}^{T-1} \mathbf{1}_n^\TRANS (c p(k) +  u(k))
\]
where the dynamics are given by \eqref{eq:TransNNs-virus} with  $w_{ij}^k$ there replaced by $m_{ij}^k(u_i(k))
$.
Equivalently, we can use the TransNN model \eqref{eq:TransNN} as the dynamic model 
\begin{equation} \label{eq:mm_si}
	s_i(k+1)=\sum_{j\in {N}_i^{\circ k}} \Psi(m_{ij}^{k}(u_i(k)), s_j(k)), \quad i \in [n],
\end{equation}
with $
\Psi(w,x)= -\log(1-w+we^{-x})	
$
 and the cost is then equivalently given by 
\[
J_2 = \sum_{k=0}^{T-1} \mathbf{1}_n^\TRANS (c (\mathbf{1}_n - \text{exp}_{\circ}(-s(k))) +  u(k)),
\]  
where $\text{exp}_\circ(-s(k))\triangleq [\exp(-s_1(k)) \dots \exp(-s_n(k))]^\TRANS $ denotes the element-wise exponential function.  The state space is $[0,\infty]^n$ and the space of control actions  is $\{0,1\}^n$. Following~\cite{ShuangPeterTransNNControl25}, we proceed to solve the optimal control problem using the Minimum Principle~\cite{bertsekas2012dynamic}. Firstly, we introduce the relaxed optimal control problem with relaxed control actions in $[0,1]^n$ (i.e. $u_i(k) \in [0,1]$ for all $i \in [n]$ and for all $k\geq 0$).  
Then the Hamiltonian for the relaxed control problem is given by
\begin{equation}\label{eq:Hamiltonian}
\begin{aligned}
	H(k)&  = ~ \mathbf{1}_n^\TRANS \left( c (\mathbf{1}_n - \exp_{\circ}(-s(k))) + u(k) \right) \\
	& + \sum_{i=1}^{n} \lambda_i(k+1) \sum_{j\in {N}_i^{\circ k}} \Psi(m_{ij}^{k}(u_i(k)), s_j(k))
\end{aligned}
\end{equation}
and hence the adjoint dynamics are given by 
\begin{equation}\label{eq:mm-lambdai}
\begin{aligned}
	\lambda_i(k)  
	& = c  e^{-s_i(k)}  \\
	& +  \sum_{\ell \in \mathcal{N}_{i-}^{\circ k} } \lambda_\ell(k+1)    \frac{\partial}{\partial s_i}   \Psi(m_{\ell i}^{k}(u_\ell(k)), s_i(k)) 
\end{aligned}
\end{equation}
with $ \lambda_i(T) = 0$, where  $\mathcal{N}_{i-}^{\circ k}: =\{\ell: (\ell, i)\in E^k\} \cup i$ denotes the outgoing neighborhood of $i$ with itself included, $m_{ij}^{k}(u_i(k))$ is given by \eqref{eq:controlled-probability}, and  
\[
\begin{aligned}
	\frac{\partial}{\partial s_i}  &\Psi( m_{\ell i}^{k}(u_\ell(k)), s_i(k))  =\\
	& \frac{m_{\ell i}^k (u_\ell(k)) e^{-s_i(k)}}{1 - m_{\ell i}^k (u_\ell (k))  + m_{\ell i}^k (u_\ell (k))  e^{-s_i(k)}} \in [0, m_{\ell i}^k (u_\ell(k))].
\end{aligned}
\]
The optimal control action can be checked node by node  without loss of generality since control actions of different nodes affect the Hamiltonian in \eqref{eq:Hamiltonian} in a decoupled way.
Moreover, it is easy to verify that the Hamiltonian $H(k)$ in \eqref{eq:Hamiltonian} is convex in $u_i(k)$ for all $i\in [n]$ and $k\geq 0$. Thus, the optimal relaxed  control action exists and  lies either on the boundary or in the interior.
Since the original control action  $u_i(k) \in \{0,1\}$ is binary, an (approximate) optimal control is determined by 
\begin{equation}
\Delta H_i(k) = H(k) |_{u_i(k)=1} - H(k) |_{u_i(k)=0}
\end{equation}
which is explicitly given by
\begin{equation*}
\Delta H_i(k) = 1 - \sum_{j\in {N}_i^{\circ k}} \lambda_j(k+1) \log \frac{1 - w_{ij}^k \beta + w_{ij}^k \beta e^{-s_j(k)}}{1 - w_{ij}^k + w_{ij}^k e^{-s_j(k)}}.
\end{equation*}
We note that (i) the $\log$ term in the equation above is negative and (ii) the adjoint process is non-negative based on its dynamics. Thus, $ \Delta H_i(k)$ could be either positive or negative, and the (approximate) optimal control rule that  minimizes the Hamiltonian can be  given by
\begin{equation} \label{eq:optimal-control}
u_i^*(k) =
\begin{cases}
1, &   \Delta H_i(k)  <  0  \\
0, & \text{otherwise}.
\end{cases}
\end{equation}

 To verify that the binary control actions generated using \eqref{eq:optimal-control} actually minimize the Hamiltonian, we compare the gradient evaluations $\frac{\partial H(k)}{\partial u_i(k)}|_{u_i(k)=0}$ and $\frac{\partial H(k)}{\partial u_i(k)}|_{u_i(k)=1}$ along the trajectories of states and adjoint states generated using the control \eqref{eq:optimal-control}: if both gradients share the same sign, the convexity of the Hamiltonian $H(k)$ in  $u_i(k)$ implies that the control action that minimizes the Hamiltonian lies on the boundary and hence given by \eqref{eq:optimal-control}.
 
 The optimal control generated above only provides candidates of optimal control solutions since the Minimum Principle only provides necessary conditions for optimality. To demonstrate that the solution is globally optimal, additional verification procedures are required. 

\subsection{Receding Horizon Control based on TransNNs} \label{subsec:rhc}

In the receding horizon control formulation, we  use  TransNNs  to predict the approximate  probability of infection in the future,  and update the current state and the current control each time  an observation of the current state $p(t) = p^0(t)$ (or $s(t) = s^0(t)$ with $s^0_i(t)  = -\log(1-p_i^0(t)) $)  becomes available at time $t$.
The corresponding receding horizon optimization problem (with shrinking horizons) is then formulated as follows: 
\[
\begin{aligned}
	J_3(p^0(t),t) %
	& = \sum_{k=t}^{T-1} \mathbf{1}_n^\TRANS (c p(k) +  u(k)),  
\end{aligned}
\]
subject to the dynamics given by \eqref{eq:TransNNs-virus} with $p(t) = p^0(t)$. At each time step,  the optimal control problem  over the horizon $\{t,t+1,\dots, T\}$ is solved (via the Minimum Principle) and only the resulting control action $u(t)$ at time $t$ is applied. The state for  node $i$ at time $t$ satisfies that $p_i(t) =1$ if infected and $p_i(t)=0$ otherwise. 

Equivalently, we can use the TransNN model \eqref{eq:TransNN} as the  model for dynamics
\begin{equation}
	s_i(k+1)=\sum_{j\in {N}_i^{\circ k}} \Psi(m_{ij}^{k}(u_i(k)), s_j(k)), \quad i \in [n],
\end{equation}
with $
\Psi(w,x)= -\log(1-w+we^{-x})	
$ and  $ s(t) = s^0(t)$ and $s^0_i(t)  = -\log(1-p_i^0(t)) $) for $i \in [n]$.
Then the corresponding receding horizon cost  is equivalently given by 
\[
J_3 (s^0(t), t) = \sum_{k=t}^{T-1} \mathbf{1}_n^\TRANS (c (\mathbf{1}_n - \text{exp}_{\circ}(-s(k))) +  u(k)),
\]  
where $\text{exp}_\circ(\cdot)$ denotes the element-wise exponential function.
The solution involves solving the following joint equations  over the horizon $\{t, t+1,\dots, T\}$ for all $i \in [n]$
\begin{equation}
	s_i(k+1)=\sum_{j\in {N}_i^{\circ k}} \Psi(m_{ij}^{k}(u_i^*(k)), s_j(k))\end{equation}
with $s_i(t) = s_i^0(t)$
and  
\begin{equation}\label{eq:oc-lambda_i}
\begin{aligned}
	\lambda_i(k)  
	& = c  e^{-s_i(k)}  \\
	& +  \sum_{\ell \in \mathcal{N}_{i-}^{\circ k} } \lambda_\ell(k+1)    \frac{\partial}{\partial s_i}   \Psi(m_{\ell i}^{k}(u_\ell(k)), s_i(k)) 
\end{aligned}
\end{equation}
with $ \lambda_i(T) = 0$, where  $\mathcal{N}_{i-}^{\circ k}: =\{\ell: (\ell, i)\in E^k\} \cup i$ denotes the outgoing neighborhood of $i$ with itself included. The optimal control rule used in solving the joint equations is  given by
\begin{equation}
u_i^*(k) =
\begin{cases}
1, &   \Delta H_i(k)  <  0  \\
0, & \text{otherwise}
\end{cases}
\end{equation}
with 
\begin{equation*}
\Delta H_i(k) = 1  \textcolor{black}{-\lambda_i(k+1)}\sum_{j\in {N}_i^{\circ k}}  \log \frac{1 - w_{ij}^k \beta + w_{ij}^k \beta e^{-s_j(k)}}{1 - w_{ij}^k + w_{ij}^k e^{-s_j(k)}}
\end{equation*}
for all $k\in \{t, t+1, \dots, T-1\}$. 
At time $t$, we only apply the first control action $u_i^*(t)$ for all $i \in [n]$.
 
The state for  node $i$ observed at time $t$ satisfies that $s_i(t) =+\infty$ if infected and $s_i(t)=0$ otherwise. 
We note that $s(t)$ may contain $+\infty$ at the (updated initial) time $t$, but the cost function and the states $s(k)$ for $k>t$ always remain bounded. 

This receding horizon control scheme allows us to update the observations at time $t$ instead of using the evolution of probabilities given by the TransNN model over the entire horizon $\{0,1, ..., T\}$ specified in Section \ref{subsec:oc-TransNNs}.
The form of solution is not exactly dynamic programming as in the MDP solution (in Section \ref{subsec:MDP-SIS}) as we use the TransNN dynamics instead of the  Markovian dynamics with $2^n$ state configurations, and it is not exactly following the Minimum Principle for optimal control over the horizon $\{0,1, ..., T\}$ (in Section \ref{subsec:oc-TransNNs}) as we keep updating the horizon $\{t, t+1,\dots,  T\}$ as we observe the new infection state at time $t$.

\subsection{Complexity}

Compared to the MDP solution, the advantage of the receding horizon control  is that there is no need to compute the value function over the $2^n$ state configurations.  
Compared to the optimal control based on TransNNs, the advantage of the receding horizon control is that the resulting control is less conservative as we update the state observation over time, and  its disadvantage is that $T$ number of optimal control problems instead of $1$ optimal control problem need to be solved, over the horizon $\{0,1,\dots,  T\}$.

For naive MDP solutions, at each dynamic programming step,  $\mathcal{O}( 2^{3n}  {n d})$  elementary operations are needed since the number of  state configurations and that of control actions are both $2^n$ and each evaluation of \eqref{eq:MarkovTranProb2} requires $\mathcal{O}({n d})$ elementary operations, where $\mathcal{O}$ is notation for the standard Big-O asymptotic upper bound and  $d$ denotes the maximum (incoming and outgoing) degree over the physical contact graph $([n], E^k)$ for all $0 \leq k \leq T$. Since there are in total  $T$ dynamic programming steps, the complexity of the MDP solution is $\mathcal{O}(2^{3n} T {n d} )$, and hence is exponential in $n$ and linear in $T$ (which is also demonstrated later by the numerical results  in Fig.~\ref{fig:compute_time}). 

For the optimal control based on TransNNs, the solution involves the computations of the  adjoint dynamics, the state dynamics and control actions, each of which requires $\mathcal{O}(n T d )$  elementary operations where $d$ denotes the maximum (incoming and outgoing) degree over the physical contact graph $([n], E^k)$ for all $0 \leq k \leq T$.  Let $n_{\text{iter}}$ denote the number of iteration needed  for solving the joint equations \eqref{eq:mm_si}, \eqref{eq:mm-lambdai} and \eqref{eq:optimal-control} (for instance, using fixed-point iterations). Then the complexity in terms of elementary operations is $\mathcal{O}(n T d n_{\text{iter}} )$. 

The complexity of the receding horizon control involves solving similar optimal control problems iteratively over the horizon $\{0,1,\dots,  T-1\}$. 
Hence it involves $\mathcal{O}(n T^2 d n_{\text{iter}})$ elementary operations.

\section{Numerical Results}
The parameters used in the numerical examples are: $\beta =0.3, ~T=10, ~c=200.$ For the clarity of illustrations, we select $n = 5$ (a much larger $n$ is possible for control based on TransNNs, but not feasible for MDP solutions). The network structure and transmission probabilities are illustrated in Fig.~\ref{fig:net-prob}. 
\begin{figure}[htb] 
\centering
	\includegraphics[width=7.5cm]{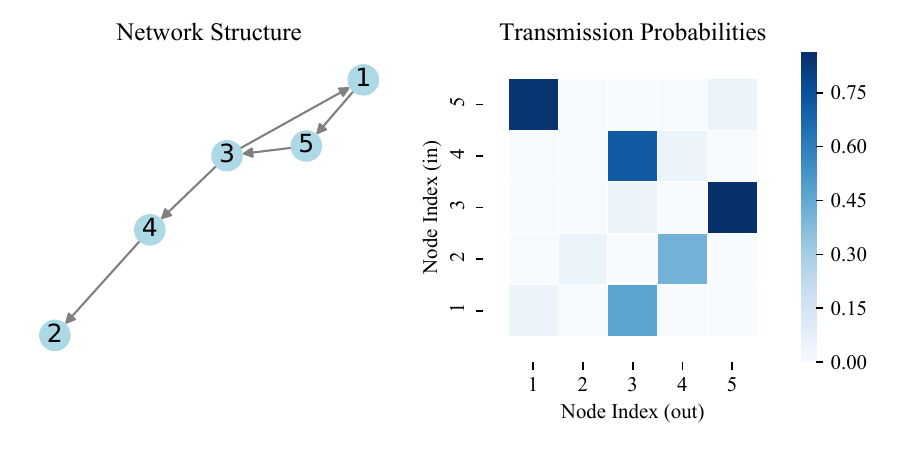}
	\caption{A network example with 5 nodes (left) and the transmission probabilities (right) among nodes.} \label{fig:net-prob}
\end{figure}
\begin{figure}[htb]
\centering
		\includegraphics[width=7.8cm]{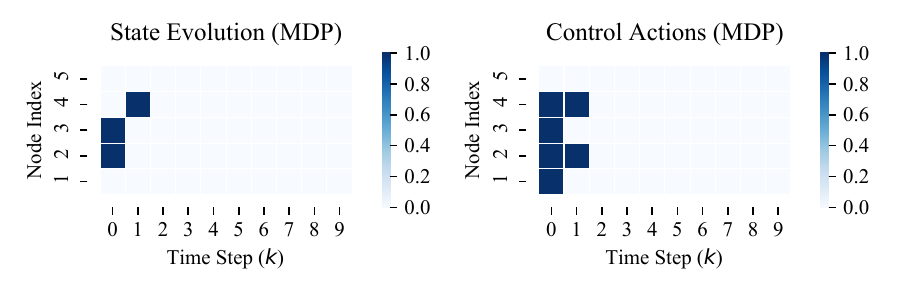}
		\caption{Control actions (right) generated from the MDP control, and actual state realizations (left) the under such control actions.}\label{fig:mdp-results}
\end{figure}
\begin{figure}[htb]
\centering
		\includegraphics[width=8cm]{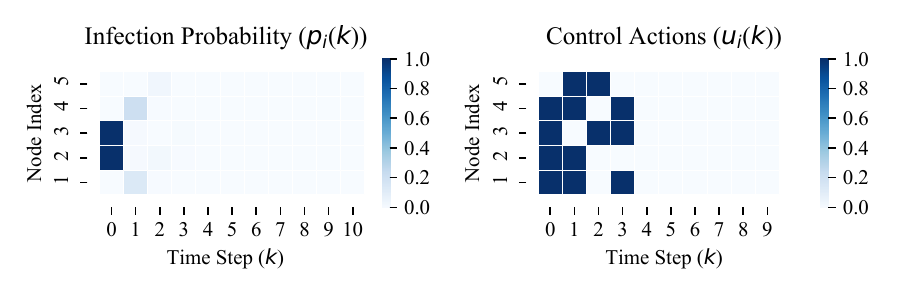}
		\caption{Control actions (right) generated from TransNN-based optimal control  and the infection probabilities (left) in the TransNN model under such control actions.}  
		\label{fig:trans-control-results}
\end{figure}
\begin{figure}[htb]
\centering
		\includegraphics[width=8cm]{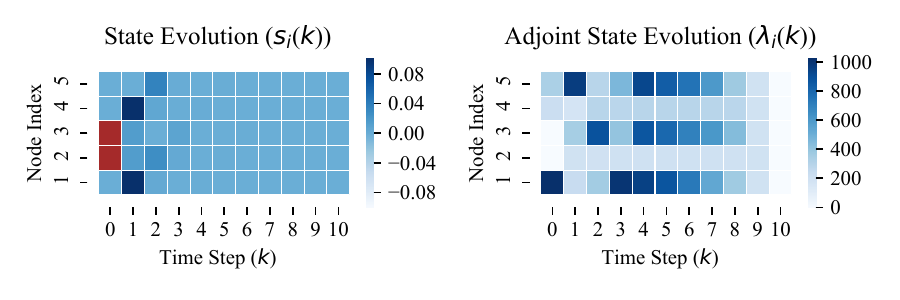}
		\caption{States of TransNNs (left)  under the TransNN-based optimal control actions and the adjoint states (right). Brown squares (left) represent $s_i(0)= \infty$ which corresponds to $p_i(0)=1$. Although $s_i(0)$ may be $+\infty$ for some $i \in [n]$,  the cost and the state $s_i(k)$ for all $k>0$ and all $i\in [n]$ always remain bounded. }
		\label{fig:info-adjoint-states}
\end{figure}
\begin{figure}[htb]
\centering
		\includegraphics[width=8cm]{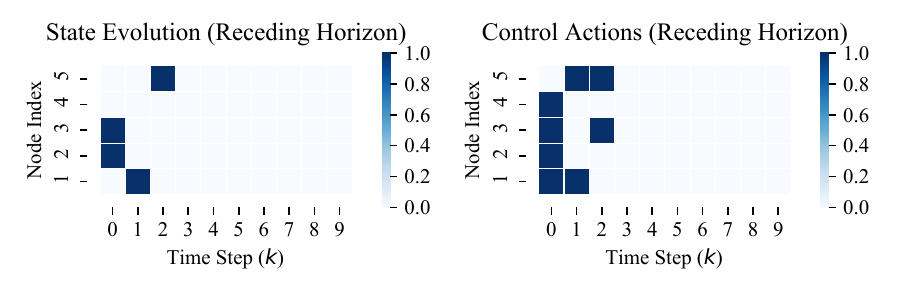}
		\caption{Control actions (right) generated from  TransNN-based receding horizon control (see Section~\ref{subsec:rhc})  and one actual state realization (left) the under such control law.}\label{fig:rhc-results}
\end{figure}

The state realizations and the control actions in one simulation under the optimal control from solving the MDP is given in Fig.~\ref{fig:mdp-results}. The execution time for solving the MDP problem is 16.72 seconds on a standard MacBookPro laptop.  
 The probability states of TransNNs and the optimal control actions generated from controlling TransNNs are shown in Fig.~\ref{fig:trans-control-results}, and the adjoint states and the (information content) states of TransNNs  are shown in Fig.~\ref{fig:info-adjoint-states}.  The execution time for solving the TransNN control problem is 0.020 seconds on the same laptop (which corresponds to a computational time reduction by about 3 orders of magnitude compared to solving MDP).  The state realizations and the corresponding control actions generated from  the receding horizon control based on TransNNs are illustrated in Fig.~\ref{fig:rhc-results}. The execution time for solving the receding horizon control is 0.094 seconds on a standard MacBookPro laptop (which corresponds to a computational time reduction by more than 2 orders of magnitude compared to solving MDP).

 From Fig.~\ref{fig:mdp-results} and Fig.~\ref{fig:trans-control-results}, we see that TransNN-based optimal control actions include all the control actions generated from MDP solutions,  and at the first time step the control actions under both solution methods are the same. From Fig.~\ref{fig:rhc-results} and Fig.~\ref{fig:trans-control-results}, we observe that the receding horizon control is less conservative in terms of control actions and all the control actions generated from receding horizon control are included in the TransNN-based optimal control actions. 

Both the optimal control and the receding horizon control are suitable for problems with a large number of nodes (e.g. for the problem with $n=100$,  the computation for receding horizon control based on TransNNs takes about 25.5 seconds and that for optimal control based on TransNNs about 4.4 seconds). In contrast, using the MDP solution is challenging since it is not feasible to directly store vectors of length $2^{100}$ for the value function on a standard MacBookPro laptop.
\begin{figure}[htb]
	\subfloat[The computation time  with respect to the problem horizon $T$.]{
		\includegraphics[width=4cm]{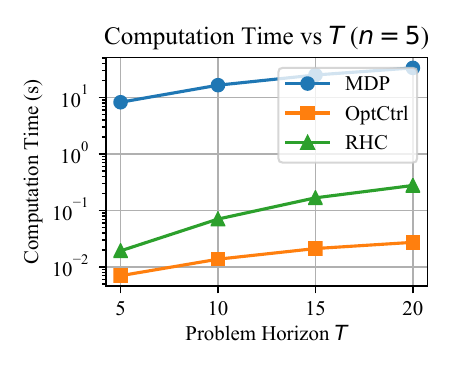}\label{fig:compute_time_vs_T}}
		\quad 
	\subfloat[The computation time with respect to the number of nodes $n$. ]{\includegraphics[width=4cm]{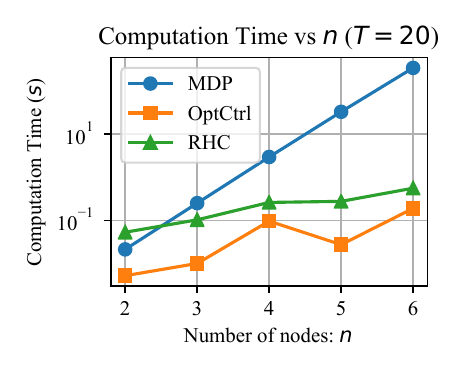}
			\label{fig:compute_time_vs_n}	}
\caption{The computation time for three control laws on a standard MacBookPro laptop: (i) MDP based on Markovian SIS dynamics, (ii) optimal control (OptCtrl) based on TransNNs  and (iii) receding horizon control (RHC) based on TransNNs. } \label{fig:compute_time}

\end{figure}

For all three control laws (i.e. MDP control, TransNN-based optimal control, and TransNN-based receding horizon control), the computation time  with respect to the problem horizon and that with respect to the number of nodes are plotted respectively in Fig.~\ref{fig:compute_time_vs_T} and Fig.~\ref{fig:compute_time_vs_n}. The computation time for naive MDP solutions increases exponentially with respect to the number of nodes and is significantly larger than control solutions  based on TransNNs.
\section{Conclusion}

This work demonstrates that TransNNs enable  approximate receding horizon control solutions for Markovian SIS dynamics with $2^n$ state configurations. 
It allows significant computational savings compared to the dynamic programming solution to Markov decision model with $2^n$ state configurations, as well as providing less conservative control actions compared to the optimal control based on TransNNs. 

Future work will (a) evaluate the receding horizon control for TransNNs with different types of control actions, (b) identify  TransNN models with immune (or inhibition) states,  (c) relax the condition (A1) to investigate the case with non-Markovian dynamics, (d) consider the cases with risk-aware control cost, (e) explore network structures to simplify MDP solutions, (f) extend the approximate receding horizon control for general MDP problems, and (g) use TransNNs as a framework to explain the computational reductions in neural networks in treating high dimension problems.

\bibliographystyle{IEEEtran}
\bibliography{mybib}

\end{document}